\documentclass[12pt,reqno]{amsart}
\usepackage[samelinetheorem,notitle]{maherart}

\usepackage{xurl}

\usepackage{amsfonts,enumerate,bm,bbm,paralist,nicefrac,comment,float,thmtools,thm-restate,comment,multirow}
\usepackage[export]{adjustbox}
\makeatletter
\newcommand\fs@boxedtop
  {\fs@boxed
   \def\@fs@mid{\vspace\abovecaptionskip\relax}%
   \let\@fs@iftopcapt\iftrue
  }
\makeatother
\floatstyle{boxedtop}
\floatname{framedbox}{Procedure}
\newfloat{framedbox}{tbp}{lob}

\hypersetup{
    pdftitle =
        {Censorship Resistance in On-Chain Auctions},
    pdfauthor =
        {Elijah Fox, Mallesh M Pai and Max Resnick},
    pdfsubject=
        {Censorship Resistance in On-Chain Auctions}
}

\usepackage{xparse}

\DeclareDocumentCommand\Pr{ m g }{%
    \ensuremath{   \IfNoValueTF {#2}
      {\mathbb{P}\left[{#1}\right]}
      {\mathbb{P}\left[{#1}\middle\vert{#2}\right]}%
    }
}
\DeclareDocumentCommand\E{ m g }{%
    \ensuremath{   \IfNoValueTF {#2}
      {\mathbb{E}\left[{#1}\right]}
      {\mathbb{E}\left[{#1}\middle\vert{#2}\right]}%
    }
}

\newif\ifdraft
\drafttrue

\begin{document}
\raggedbottom
\title{Contingent Fees in Order Flow Auctions}

\author[Resnick]{Max Resnick$^\text{A}$ }
\address{$^\text{\MakeLowercase{a}}$Rook Labs,
\href{mailto:max@riskharbor.com}{max@rook.fi}}

\address{\today}
\thanks{Resnick gratefully acknowledges support from a Flashbots grant which funded this work. Comments from Quintus Kilbourn and Mallesh Pai were greatly appreciated}
\begin{abstract}
Many early order flow auction designs handle the payment for orders when they execute on the chain rather than when they are won in the auction. Payments in these auctions only take place when the orders are executed, creating a free option for whoever wins the order. Bids in these auctions set the strike price of this option rather than the option premium. This paper develops a simple model of an order flow auction and compares contingent fees with upfront payments as well as mixtures of the two. Results suggest that auctions with a greater share of the payment contingent on execution have lower execution probability, lower revenue, and increased effective spreads in equilibrium. A Reputation system can act as a negative contingent fee, partially mitigating the downsides; however, unless the system is calibrated perfectly, some of the undesirable qualities of the contingent fees remain. Results suggest that designers of order flow auctions should avoid contingent fees whenever possible.
\newline
\newline
\noindent \textsc{Keywords:} auctions, blockchain, order flow, on-chain, options
\newline
\newline

\end{abstract}


\maketitle

\thispagestyle{empty}

\newpage

\pagenumbering{arabic}

Suppose you sit down at a restaurant where you only pay for what you eat. How much food do you order? How much goes to waste? You could read the menu and select whichever dish sounds the best. But, wouldn't you rather order every item on the menu and see which looks the most appetizing before you make your decision? In that case, every other dish, having already been prepared by the kitchen, would be wasted, but you would be slightly happier since you had more information when selecting the meal. 

This pay-what-you-eat mechanism is more commonly known as a contingent fee. Payments contingent on information revealed after the sale of the item are useful for a few reasons, most commonly when the buyer is budget constrained initially or when the buyer is risk-averse. For example, lawyers often charge legal fees only if they win \cite{10.2307/2555962, clermont1977improving, miceli1991contingent} and oil site auctions commonly involve both an upfront payment and a contingent fee \cite{haile2010recent,hendricks1988empirical,hendricks1993optimal}, due only if the company finds oil. But, MEV searchers are neither particularly risk-averse nor budget constrained. They typically rake in significant profits, which could easily be reinvested into upfront costs for order flow. Moreover, the amount of money paid for order flow per slot is likely small relative to total profits, suggesting that execution variance washes over longer time periods.

Contingent fees erode the value of the option because they make executing the option less likely to be profitable. An upfront payment is a sunk cost so it is just as profitable to execute an order that you paid 100 dollars upfront for as it is to execute an order that you paid 1 dollar up front for, \textit{ceteris paribus}. But, once you add a contingent fee, executing the option has to not only be profitable but profitable enough to make up for the contingent fee.

When order flow auctions have high contingent fee components as a portion of the bid amount, increasing the bid devalues the option, by increasing the contingent fee. This can lead to a race to the bottom where bidders raise bids, in order to win the order flow, but in the process increase the contingent fee to a point where the order is almost never worth it to execute. 

I demonstrate this with a simple model where a single piece of order flow is auctioned off. Bidders observe the order and place bids. After the winner is chosen, new information is revealed that affects the value of executing the order (the price of the underlying asset is updated). The winning bidder then chooses whether to execute the order or not based on this information. 

I compare the equilibria in this model under purely contingent and purely upfront payment regimes. In the purely contingent regime, in a perfectly competitive auction, the equilibrium bid is high enough that the option to execute is executed almost never, and, as a consequence, all revenue vanishes. In the upfront regime revenue is equal to the expected value of the option, and there is no impact on execution probability. Mixtures of upfront and contingent payments favor the mixture with a greater upfront share on execution probability, revenue, and effective spread metrics. Results are similar when the model is extended to allow some probability of execution or failure without input from the winning bidder, as long as the winning bidder chooses to execute with positive probability. 

\section{Model}
Suppose a single transaction $t$ is being auctioned off. For simplicity assume it is a limit sell order. It follows that the right to execute the order is an option to buy whatever the transaction is selling. The value of the transaction to searchers in the order flow auction will depend linearly on the behavior of an external reference price $S_1$ which is only revealed after the auction has concluded. In particular, the value of the transaction, conditional on execution, as a function of $S_1$ is given by
\begin{equation*}
    S_1 - K
\end{equation*}
Where $K$ is the strike price of the limit order $t$. 
Suppose that $S_1$ has compact connected support. 

We will consider two similar auction styles, one with unconditional payment on winning the auction and the other where payments occur if and only if the transaction is actually executed. 

The unconditional mechanism works as follows:
\begin{enumerate}
    \item Bidders submit bids $b_1,\dots,b_n$.
    \item The bidder who submitted the highest bid wins and pays their bid $b$
    \item The winning bidder observes $S_1$ and then decides whether or not to execute the transaction.
\end{enumerate}
The winning bidder's payoff is given by:
\begin{equation*}
    V_i = \begin{cases}
        S_1 - K - b& \text{Executed}\\
        -b & \text{Not Executed}
    \end{cases}
\end{equation*}
And all other bidders receive nothing and pay nothing.

The conditional mechanism is similar except that the winning bid is paid only if the transaction executes:
\begin{enumerate}
    \item Bidders observe $S_0$ and submit bids $b_1,\dots,b_n$
    \item The bidder who submitted the highest bid wins.
    \item The winning bidder observes $S_1$ and then decides whether or not to execute the transaction. If the transaction is executed, he will pay his bid $b$. 
\end{enumerate}
The winning bidder's payoff is given by:
\begin{equation*}
    V_i = \begin{cases}
       S_1 - K - b & \text{Executed}\\
        0 & \text{Not Executed}
    \end{cases}
\end{equation*}
\subsection{Extensions}
We extend this model to mixtures of the two games by introducing a parameter $\alpha \in [0,1]$, that represents the portion of the winning bid which is paid regardless of execution. $(1-\alpha)b$ is paid when the order is executed. This model is particularly salient in the case of external penalties other than payments that correspond to execution failure, for example, many projects have reputation systems that penalize searchers who don't execute transactions; however, these reputation systems can be difficult to tune and may not penalize searchers enough (or might penalize them too much in some cases). The winning bidder's payoff in this extension as a function of $\alpha$ is given by:
\begin{equation*}
    V_i = \begin{cases}
        S_1 - K - b & \text{Executed}\\
        \alpha b & \text{Not Executed}
    \end{cases}
\end{equation*}
We add a second extension that models scenarios in which searchers may not have full control over the execution of a transaction if they win. We introduce parameters $p$ and $q$ where $p$ represents the probability that the transaction executes regardless of the winning searcher's input and $q$ is the probability that the transaction fails regardless of the winning searcher's input. $1-p-q$ represents the probability that execution is decided by the searcher.

\section{Results}
Throughout we consider competitive equilibria of the auctions, that is we assume that searchers compete away all profits, leading to the zero profit condition. Although these games look similar, their competitive equilibria could not be more different. Beginning with the upfront payment game. Notice that, by the zero profit condition, the winning bid will be equal to the expected value of the option:
\begin{equation*}
    \mathbb{E}\max(S_1 - K, 0)
\end{equation*}
But in the contingent fee game, if we compute the value of winning the auction, we see that it depends on $b$. The winning bidder will execute only when 
\begin{equation*}
    S_1 - K > b \iff S_1 - (K + b) > 0.
\end{equation*}
So the payoff for winning the auction in the contingent fee game is 
\begin{equation*}
    \max(S_1 - K -b, 0).
\end{equation*}
Again, the winning bidder's payoff should satisfy the zero profit condition. But,
\begin{align*}
&\mathbb{E}[\max(S_1-K-b,0)] = 0\\
\iff &P(S_1-K-b \geq 0) = 0\\
\iff &\max(\text{Support}(S_1)) - K \leq b
\end{align*}
So, in equilibrium, the second-highest bid is so high that it erodes the value of the order flow completely and it is never profitable for the winner to execute the transaction. We summarize these results in theorem \ref{thm:corner_cases}:

\begin{theorem}\label{thm:corner_cases} In the \textbf{upfront payment} game, the competitive equilibrium bid is just the value of the option $b = \mathbb{E} \max(S_1-K,0)$. The probability of execution is $P(S_1 > K)$ and the expected revenue is $P(S_1 > K)\mathbb{E} \max(S_1-K,0)$.

In the \textbf{contingent fee} game, the competitive equilibrium bid is $b = \max(\text{Support}(S_1))- K$. The probability of execution is $0$ and the expected revenue is $0$. 

\end{theorem} 
\subsection{Extensions}
We now consider the case of the mixture of the two games. The value of the exclusive right to execute the order in period $1$ can be thought of as a call option.

\begin{equation*}
    C_K(S_1) = I_{S_1 \geq K}(S_1 - K),
\end{equation*}
where $I_{s_1 \geq K}$ is the indicator function of $s_1 \geq K$. The expected value is then:
\begin{align*}
    \mathbb{E}[C_K(S_1)] &= P(S_1 \geq K)(\mathbb{E}[S_1|S_1 > K] -K)\\
    &= (1-F(K))\left(\int_K^\infty \frac{f(x)}{1-F(K)}dx - K\right)
\end{align*}
Now we consider a setting where the winning bidder, denoted $w$ pays $\alpha b$ regardless of execution and an additional $(1-\alpha)b$ upon execution. This has the effect of shifting the strike price of the option by $(1-\alpha)b$. The winner's expected utility is then given by:
\begin{equation}\label{eq:1}
    \mathbb{E}[U_w] = (1-F(K+(1-\alpha)b))\left(\int_{K+(1-\alpha)b}^\infty \frac{f(x)}{1-F(K+(1-\alpha)b)}dx - K-(1-\alpha)b\right) - \alpha b.
\end{equation}
By the zero profit condition, the equilibrium bid is the solution to:
\begin{equation}\label{eq:2}
    \alpha b = (1-F(K+(1-\alpha)b))\left(\int_{K+(1-\alpha)b}^\infty \frac{f(x)}{1-F(K+(1-\alpha)b)}dx - K-(1-\alpha)b\right) .
\end{equation}
Notice that \eqref{eq:1} is positive for $b=0$ and eventually negative as $b \to \infty$ by the law of total probability. Therefore we can use a bracketing zero finder to solve \eqref{eq:2}. 

We will denote the solution for b in \eqref{eq:2} as $b^*$ then we can write the probability of execution as:
\begin{equation}\label{eq:P_execution}
    P(\text{Execution}) = (1-F(K+(1-\alpha)b^*(\alpha))
\end{equation}

The effective spread is given by the expectation of $S_1 - K$ conditional on the order executing. The conditional PDF of $S_1$ given $S_1 > K + (1-\alpha)b^*(\alpha)$ is:
\begin{equation*}
    f_{S_1|S_1 > K + (1-\alpha) b^*} = \frac{f(x)}{1-F(K+(1-\alpha)b^*(\alpha))},
\end{equation*}
for $S_1 > K + (1-\alpha)b^*(\alpha)$ and $0$ otherwise. Multiplying by the effective spread and integrating we get
\begin{align*}
    E[S_1 -K|S_1 > K + (1-\alpha)b^*(\alpha)] &= \int_{K+(1-\alpha)b^*(\alpha)}^\infty (x-K) f_{S_1|S_1 > K + (1-\alpha) b^*}(x) dx\\
    &= \dfrac{\int_{K+(1-\alpha)b^*(\alpha)}^\infty (x-K)f(x)dx}{1-F(K+(1-\alpha)b^*(\alpha))}
\end{align*}
The auction's revenue is given by the up-front payment $\alpha  b^*(\alpha)$ and the contingent fee $(1-\alpha)b^*(\alpha)$ multiplied by the probability of execution $1 - F(K+(1-\alpha)b^*(\alpha))$.
\begin{equation*}
    \alpha b^*(\alpha) - (1-F(K + (1-\alpha)b^*(\alpha))(1-\alpha)b^*(\alpha)
\end{equation*}
Results for the uniform case are availible in the apendix, section \ref{sec:apendix}, and visual representations can be seen in figure \ref{fig:uniform}. Figure \ref{fig:Beta} displays results for several parametrizations of the Beta distribution. 

\begin{figure}
    \centering
    \includegraphics[width = \textwidth]{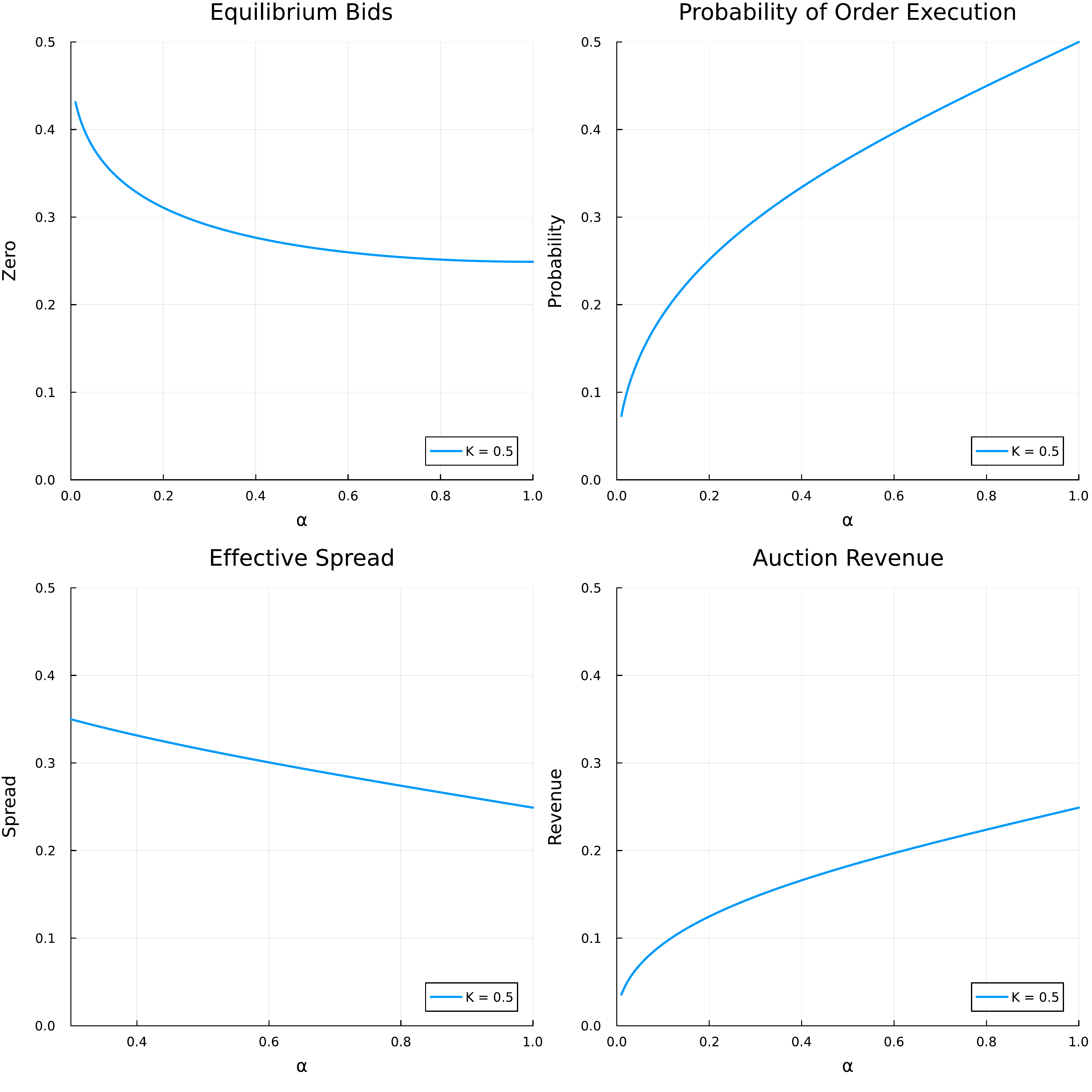}
    \caption{Results for $S_1 \sim U[0,1]$, $K = .5$ }
    \label{fig:uniform}
\end{figure}
\begin{theorem}
The following hold
\begin{enumerate}
    \item Probability of order execution is monotonically increasing in $\alpha$.
    \item Auction revenue is monotonically increasing in $\alpha$.
    \item Effective spread is monotonically decreasing in $\alpha$.
\end{enumerate}
\end{theorem}

\begin{proof}
(1) follows immediately from \ref{eq:P_execution}
\begin{equation*}
    P(\text{Execution}) = (1-F(K+(1-\alpha)b^*(\alpha))
\end{equation*}
Note that $(1-\alpha)$ is decreasing in $\alpha$ and $b^*(\alpha)$ is decreasing in $\alpha$, therefore $K+(1-\alpha)b^*(\alpha)$ is decreasing in $\alpha$. So $F(K+(1-\alpha)b^*(\alpha))$ is decreasing in $\alpha$ because $F$ is a cumulative density function and is therefore increasing. Inverting one final time we have that $1- F(K+(1-\alpha)b^*(\alpha))$ is increasing in $\alpha$. 

(2) follows from the zero profit condition. The amount the winning bidder pays in expectation must be equal to the value of the option that they receive. Since the contingent fee erodes the value of the option in proportion to $b^*(\alpha)$ which is decreasing in $\alpha$ revenue must be increasing in $\alpha$. 

(3) is a result of eliminating a smaller region of $S_1$'s support closest to $K$ as $\alpha$ increases. The result is that the average execution quality increases in $\alpha$.
\end{proof}
\begin{figure}[b]
    \centering
    \includegraphics[width = \textwidth]{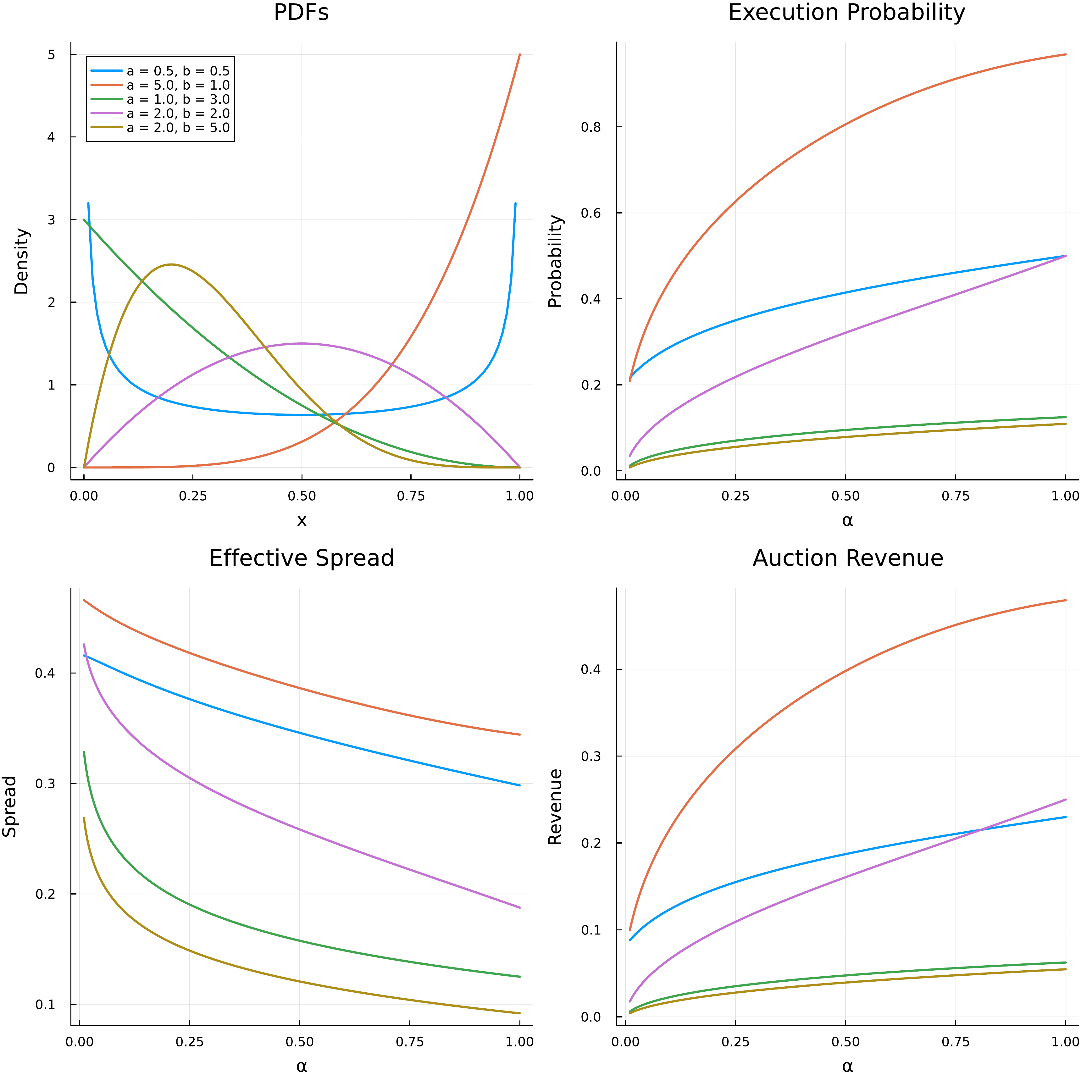}
    \caption{Results for $S_1 \sim \mathbb{\beta}[a,b]$, $K = .5$ }
    \label{fig:Beta}
\end{figure}

\subsection{Extension to involuntary execution and failure}
Now we extend the model to a situation where with probability $p$, the execution takes place automatically, with probability $q$ execution fails, and with probability $1-p-q$ the execution is optional (chosen by the winning bidder). In this scenario, the expected utility for the winner is given by:
\begin{align*}
    \mathbb{E}[U_w] &= p( \mathbb{E}[S_1] - K -(1-\alpha) b)\\
    &+ (1-p-q)(1-F(K+(1-\alpha)b))\left(\int_{K+(1-\alpha)b}^\infty \frac{f(x)}{1-F(K+(1-\alpha)b)}dx - K-(1-\alpha)b\right)\\
    &- \alpha b
\end{align*}
This can also be solved using a bracketing zero finder. Letting $b^*$ denote the solution as before, we can write the probability of execution as 
\begin{equation*}
    P(\text{Execution}) = p + (1-p-q)(1-F(K+(1-\alpha)b^*(\alpha))
\end{equation*}
But with $p+q < 1$ This reduces to the model in the previous section by modifying $F$, and $\alpha$. The modification procedure for $F$ involves inserting a mass of probability $q$, the probability of automatic failure, below the strike price $K$ then normalizing the distribution by dividing through by $1+q$. The modification procedure for $\alpha$ is to set $\alpha' = \alpha + p(1-\alpha)$. 

Therefore the insights of the previous section hold here as long as $p+q < 1$. 
\section{Discussion}

This paper primarily examines MEV order flow auction payment structures. Reputation systems, acting as negative contingent fees where reputational damage rather than monetary loss results from not executing orders, could also improve auction performance. We discuss the potential and limitations of reputation systems in MEV auctions.

Reputation systems can incentivize participants to promote market efficiency and stability by penalizing undesirable behaviors. Participants may be more inclined to execute orders when reputation is at stake, even with economically unattractive contingent fees.

However, reputation systems have limitations:

\textbf{Manipulation:} Susceptible to strategic behaviors that inflate reputation scores, such as collusion, fake accounts, or exploiting loopholes, which could undermine the system's efficacy.

\textbf{Complexity:} Designing and implementing an accurate, transparent, fair, and manipulation-resistant reputation system can increase the overall auction mechanism complexity.

\textbf{Enforcement:} Enforcing reputational penalties is challenging, as participants in decentralized environments or markets with low entry barriers may re-enter under new identities, circumventing reputational consequences.

\textbf{Trade-offs:} Introducing reputation systems creates trade-offs in auction design, such as reduced competition or increased barriers to entry for new participants, potentially offsetting benefits from improved execution rates and reduced effective spreads.

In summary, reputation systems offer potential benefits but have limitations. MEV auction designers should weigh benefits against challenges and trade-offs. Further research could explore optimal reputation system design in MEV order flow auctions and investigate how these systems interact with payment structures to achieve desired outcomes.

\subsection{Conclusion}
This straightforward model of order flow auctions reveals the impact of varying payment structures—contingent fees, upfront payments, and their mixtures—on key outcomes. Findings indicate that higher contingent payments result in lower execution probability, decreased revenue, and increased effective spreads.

These insights hold significant implications for order flow auction design, suggesting a focus on upfront payments rather than contingent fees. Overemphasizing contingent fees risks inciting a race to the bottom, adversely impacting users and the order flow provider.

Future research could expand the model by incorporating multiple orders or multiple blocks to execute and explore more realistic assumptions about bidder risk aversion and budget constraints.

\newpage

\bibliographystyle{aer}
\bibliography{onchain-auctions}

@article{hendricks1988empirical,
  title={An empirical study of an auction with asymmetric information},
  author={Hendricks, Kenneth and Porter, Robert H},
  journal={The American Economic Review},
  pages={865--883},
  year={1988},
  publisher={JSTOR}
}

@article{hendricks1993optimal,
  title={Optimal selling strategies for oil and gas leases with an informed buyer},
  author={Hendricks, Kenneth and Porter, Robert H and Tan, Guofu},
  journal={The American Economic Review},
  volume={83},
  number={2},
  pages={234--239},
  year={1993},
  publisher={JSTOR}
}

@article{haile2010recent,
  title={Recent US offshore oil and gas lease bidding: A progress report},
  author={Haile, Philip and Hendricks, Kenneth and Porter, Robert},
  journal={International Journal of Industrial Organization},
  volume={28},
  number={4},
  pages={390--396},
  year={2010},
  publisher={Elsevier}
}

@article{10.2307/2555962,
 ISSN = {07416261},
 URL = {http://www.jstor.org/stable/2555962},
 author = {Daniel L. Rubinfeld and Suzanne Scotchmer},
 journal = {The RAND Journal of Economics},
 number = {3},
 pages = {343--356},
 publisher = {[RAND Corporation, Wiley]},
 title = {Contingent Fees for Attorneys: An Economic Analysis},
 urldate = {2023-03-31},
 volume = {24},
 year = {1993}
}

@article{clermont1977improving,
  title={Improving on the contingent fee},
  author={Clermont, Kevin M and Currivan, John D},
  journal={Cornell L. Rev.},
  volume={63},
  pages={529},
  year={1977},
  publisher={HeinOnline}
}

@article{miceli1991contingent,
  title={Contingent fees for lawyers: the impact on litigation and accident prevention},
  author={Miceli, Thomas J and Segerson, Kathleen},
  journal={The Journal of Legal Studies},
  volume={20},
  number={2},
  pages={381--399},
  year={1991},
  publisher={The University of Chicago Law School}
}

\newpage

\section{Closed from solution for the uniform case}\label{sec:apendix}
Letting $K = 1/2$ and $S_1 \sim U[0,1]$ for the purpose of exposition, we can solve \eqref{eq:1} explicitly for $b^*$
\begin{align*}
\alpha b& = (1-F(K+(1-\alpha)b))\left(\int_{K+(1-\alpha)b}^\infty \frac{f(x)}{1-F(K+(1-\alpha)b)}dx - K-(1-\alpha)b\right)\\
\alpha b& = (1-1/2-(1-\alpha)b)\left(\int_{1/2+(1-\alpha)b}^1 \frac{1}{1-1/2-(1-\alpha)b}dx - 1/2-(1-\alpha)b\right)\\
\alpha b & = 1- 1/2 -(1-\alpha)b +(1-1/2-(1-\alpha)b)(-1/2 - (1-\alpha)b)\\
\alpha b & =-2(1-\alpha)b + 1/4 +(1-\alpha)b + (1-\alpha)^2b^2\\
\alpha b & =-(1-\alpha)b + 1/4  + (1-\alpha)^2b^2\\
1/4 & = (1-\alpha)^2b^2 -(1-\alpha)b - \alpha b\\
0 &= ((1 - \alpha)^2) b^2 - b + \frac{1}{4}\\
b &= \frac{1 \pm \sqrt{1 - (1 - \alpha)^2}}{2(1 - \alpha)^2}
\end{align*}
The negative branch is the correct one here so our formula for $b^*(\alpha)$ is:
\begin{equation*}
   b^*(\alpha) = \frac{1 - \sqrt{1 - (1 - \alpha)^2}}{2(1 - \alpha)^2}
\end{equation*}
The Probability of execution is then 
\begin{equation*}
    1/2 - \frac{1 - \sqrt{1 - (1 - \alpha)^2}}{2(1 - \alpha)}
\end{equation*}
Auction revenue is
\begin{equation*}
    \alpha \frac{1 - \sqrt{1 - (1 - \alpha)^2}}{2(1 - \alpha)^2} +  (1-\alpha)\left(1/2 - \frac{1 - \sqrt{1 - (1 - \alpha)^2}}{2(1 - \alpha)}\right) \frac{1 - \sqrt{1 - (1 - \alpha)^2}}{2(1 - \alpha)^2}
\end{equation*}

\end{document}